\documentclass[9pt,preprint]{sigplanconf}

\usepackage{microtype}

\usepackage[yyyymmdd,hhmmss]{datetime}

\usepackage{amsmath, amssymb, amsthm}

\usepackage{hyperref}
\hypersetup{
    colorlinks=true,    
    linkcolor=blue,     
    citecolor=blue,     
    filecolor=magenta,  
    urlcolor=magenta,   
}
\newcommand{\aref}[2]{\hyperref[#2]{#1~\ref{#2}}}

\usepackage[numbers,sort&compress,square]{natbib}

\usepackage{graphicx}

\usepackage{environments}  
\usepackage{macros}

\begin{document}

\setlength{\pdfpageheight}{\paperheight}
\setlength{\pdfpagewidth}{\paperwidth}

\conferenceinfo{CONF 'yy}{Month d--d, 20yy, City, ST, Country} 
\copyrightyear{20yy} 
\copyrightdata{978-1-nnnn-nnnn-n/yy/mm} 
\doi{nnnnnnn.nnnnnnn}

\title{Reachability in Two-Dimensional Unary \\
       Vector Addition Systems with States is NL-Complete
\thanks{Supported by the EPSRC, grant EP/M011801/1.}}

\authorinfo{Matthias Englert \and Ranko Lazi\'c \and Patrick Totzke}
           {DIMAP, Department of Computer Science, University of Warwick}
           {\url{warwick.ac.uk/dimap/}}

\maketitle

\begin{abstract}
Blondin et al.\ showed at LICS 2015 that 
two-dimensional vector addition systems with states
have reachability witnesses of length 
exponential in the number of states and 
polynomial in the norm of vectors.
The resulting guess-and-verify algorithm is optimal (\PSPACE),
but only if the input vectors are given in binary.
We answer positively the main question left open by their work,
namely establish that reachability witnesses 
of pseudo-polynomial length always exist.  
Hence, when the input vectors are given in unary, 
the improved guess-and-verify algorithm requires only logarithmic space.

\end{abstract}

\section{Introduction} 
\label{sec:introduction}
To quote from Boja\'nczyk's preface to \citeauthor{Schmitz16}'s
very recent survey \citep{Schmitz16},
the reachability problem for vector addition systems with states (VASS)
`is one of the most celebrated decidable problems 
 in theoretical computer science'.
The interest, though, is not only theoretical:
\citeauthor{Schmitz16} has devoted a long section to
`only a small sample of the problems interreducible' 
with the reachability problem, and the domains of the problems
he identifies range over formal languages, logic,
concurrent systems and process calculi.

For informative introductions to the fascinating history of 
the VASS reachability problem, that stretches from the 1970s,
we refer the reader to \citet{Schmitz16} and \citet{BFGHM2015}.
In a nutshell, the state of the art when it comes to the problem's complexity
hinges on two recent and one old discovery:
\begin{itemize}
\item
Remarkably, \citeauthor{Lipton76}'s \EXPSPACE\ lower bound \citep{Lipton76}
is still unbeaten.
\item
The best known upper bound, by \citet{LS15}, is \emph{cubic Ackermann},
a non-primitive recursive complexity class.
\item
The largest fixed dimension for which 
an interesting upper bound is known is $2$: 
\citet{BFGHM2015} have established that 
the $2$-VASS reachability problem is in \PSPACE.
\end{itemize}

Our contribution is to resolve the main open question that arises
from the latter work, and is highlighted by \citet{Schmitz16}.  
Namely, the headline result of \citeauthor{BFGHM2015}
is that $2$-VASS reachability is \PSPACE-complete, 
but that is provided the input to the problem is succinct, 
i.e.\ the integers that specify the action, source and target vectors 
are given in binary.  When the encoding is unary,
a considerable complexity gap has remained,
between \NL\ hardness and \NP\ membership,
and that is what we close.

We believe this is noteworthy at least for the following reasons:
\begin{itemize}
\item
To make progress on the challenge of the complexity of the general problem,
it is natural to fix some parameters, especially the dimension.
Bypassing the border between dimensions $2$ and $3$,
which is where there is a jump beyond semi-linearity,
seems to be very difficult with current techniques \citep[cf.][]{BFGHM2015}.
For dimension $1$, the complexities were determined as
\NP-complete in the binary case \citep{HKOW09} and
\NL-complete in the unary case \cite{VP1975}.
\item
The unary encoding is used frequently enough,
e.g.\ the classical modeling of concurrent systems by VASS \citep{GS92} 
produces integers that are proportional to 
how many processes may interact in a single transition.
Also, VASS given in unary can be translated without blow-up to
\emph{unary VASS} whose actions contain only $-1$, $0$ and $1$,
and \citeauthor{Lipton76}'s lower bound holds already for such VASS.
\item
Our main result, that reachability for $2$-VASS in unary is in \NL,
implies the \PSPACE\ membership of the succinct variant.
Moreover, and maybe most interestingly, we obtain the \NL\ membership
by proving that $2$-VASS have reachability witnesses of
\emph{pseudo-polynomial} length, i.e.\ polynomial in the number of states
and the maximum absolute value of any action, source or target integer.
To our knowledge, this is the first time that the complexity of
an interesting restriction of the reachability problem has broken
`the size of the reachability set barrier'.
Namely, it is well-known that general VASS may have reachability sets
which are finite but Ackermannianly large \citep{cardoza76}, 
and although some researchers conjecture that the reachability problem is 
primitive recursive or even of much smaller complexity, 
the Ackermann barrier remains.
When the dimension is $2$, it is not difficult to construct 
examples with exponentially large reachability sets 
(by employing weak doubling a number of times proportional to 
the number of states---this uses integers only up to absolute value $2$),
but we prove that polynomial reachability witnesses always exist.
\item
The technique we have developed seems novel, is surprisingly involved,
and can be seen as a kind of extension of the classical $1$-dimensional 
hill cutting \citep[cf.\ e.g.][]{VP1975} to dimension $2$.
\end{itemize}

After a couple of preparatory sections, we present the main proof in 
\aref{Section}{sec:lps-short}, split into several stages.  
There, using the flattenings obtained by \citet{BFGHM2015}, we are able 
to concentrate on obtaining short reachability witnesses for $2$-VASS 
that are \emph{\LPS s}, i.e.\ without nested cycles.
We then establish consequences for arbitrary $2$-VASS in 
\aref{Section}{sec:application_to_2_vass}.

\section{On Our Marks} 
\label{sec:preliminaries}
Here we recall, fix or introduce the basic notions, notations and problems
we require.

\paragraph{Sets of Numbers.}
To restrict a set of numbers, we may write a condition in subscript,
e.g.\ $\N_{{\ge} b}$ denotes the set of all non-negative integers
that are at least~$b$.

\paragraph{Lengths, Sizes and Norms.}
We denote the length or size by single bars, 
e.g.\ the length of a word $w$ is written~$\len{w}$.

To denote the infinity norm, we employ double bars.
Thus, for a vector $\vec{v}$, $\norm{\vec{v}}$ equals
the maximum absolute value of any entry $\vec{v}_i$.
Also, for a finite set $\vec{A}$ of vectors, $\norm{\vec{A}}$ is 
the maximum of the infinity norms of its elements.

\paragraph{Rational Cones.}
We consider the cone spanned by a subset $\vec{C}$
of a $d$-dimensional rational space $\Q^d$ to be
the closure of $\vec{C}$ under addition and 
under multiplication by positive rationals.

Note that the cone of $\vec{C}$ contains the zero vector only if
it contains a line
or one of the vectors in $\vec{C}$ is zero.

\paragraph{Paths and Admissibility.}
For a finite set $\vec{A} \subseteq \Z^d$, we have that 
vectors $\vec{a} \in \vec{A}$,
finite words $\pi \in \vec{A}^*$ and
languages $L \subseteq \vec{A}^*$
induce the following reachability relations on 
the $d$-dimensional non-negative integer space~$\N^d$:
\begin{itemize}
\item
$\vec{b}\step{a}\vec{b'}$ iff $\vec{b}+\vec{a}=\vec{b'}$,
\item
${\step{\pi}} \eqdef
 {\step{\pi(1)}} ; \cdots ; {\step{\pi(\len{\pi})}}$, and
\item
${\step{L}}\eqdef\bigcup_{\pi\in L}{\step{\pi}}$.
\end{itemize}

We often refer to a word $\pi\in\sys{A}^*$ as a \emph{path},
and call the sum 
$\effect{\pi}\eqdef \pi(1) + \cdots + \pi(|\pi|)$
the \emph{effect} of $\pi$.
From a \emph{source} $\vec{s} \in \Z^d$,
the points \emph{visited} by $\pi$ are
$\vec{s} + \pi(1) + \cdots + \pi(i)$
for all $i \in \{0, \ldots, |\pi|\}$,
the last one being the \emph{target} point.
We say that $\pi$ is \emph{admissible} from $\vec{s}$ iff
$\vec{s}\step{\pi}\vec{t}$ for some $\vec{t}$, i.e.,
iff all the points visited are in~$\N^d$,
and also call $\pi$ a path \emph{from $\vec{s}$ to $\vec{t}$} in this case.

\paragraph{Vector Addition Systems and Linear Path Schemes.}
We consider a $d$-dimensional \emph{\textVASS} ($d$-\VASS) to be
a language over a finite alphabet $\vec{A} \subseteq \Z^d$
given by a non-deterministic finite automaton~$V$.

A \emph{\textLPS} (\LPS) is a special case when the language is
given by a regular expression of the form
\[
    \Lambda =
    \alpha_0
    \beta_1^*
    \alpha_1
    \beta_2^*
    \cdots
    \beta_K^*
    \alpha_K
\]
where all $\alpha_i$ and $\beta_i$ are words in $\sys{A}^*$.
We call $\beta_1$, \ldots, $\beta_K$ the \emph{cycles} of $\Lambda$.
Its \emph{length} is $\len{\Lambda}\eqdef
\len{\alpha_0
    \beta_1
    \alpha_1
    \beta_2
    \cdots
    \beta_k
    \alpha_k
}$,
and its \emph{norm} $\norm{\Lambda}$ is 
the maximum norm of any vector (i.e.\ letter) occuring in~$\Lambda$.

Restricting further, we call $\Lambda$ \emph{simple} (an~\emph{S}\/\LPS) 
when all $\alpha_i$ and $\beta_i$ are 
of length $1$, i.e., single vectors from~$\sys{A}$.

\paragraph{Paths of Linear Path Schemes.}
We regard a path of an \LPS\ as above to be given by a sequence of exponents,
i.e.\ $n_1, \ldots, n_K$ where each $n_i$ specifies 
how many times the cycle $\beta_i$ is repeated in the path.

Note that several sequences of exponents may give the same word over $\vec{A}$.
However, this non-uniqueness of representations will not cause difficulties.

\paragraph{Reachability Problems.}
These are the membership problems of the reachability relations 
that are induced by the \VASS\ and \LPS:
\begin{quote}
Given a $d$-\VASS\ $V$ (resp., \LPS\ $\Lambda$)
and vectors $\vec{s}, \vec{t} \in \N^d$, 
decide whether $\vec{s} \step{V} \vec{t}$ 
(resp., $\vec{s} \step{\Lambda} \vec{t}$).
\end{quote}

\noindent
There are two variants of the problems: 
\emph{unary} and \emph{binary}, depending on how the integers in 
$V$ (resp., $\Lambda$), $\vec{s}$ and $\vec{t}$ are encoded.

\section{Get Set}
We have six lemmas here that are useful in the sequel.
The first four are essentially simple consequences in the plane of
Cramer's Rule and Farkas-Minkowski-Weyl's Theorem.

From Cramer's Rule, we get that for cones that contain the zero vector,
the latter is expressible using at most three vectors from the spanning set,
moreover with small positive coefficients:

\begin{lemma}
    \label{l:small.zero}
    If the cone of $\vec{C} \subseteqfin \Z^2$
    contains $\vec{0}$, then
    $\vec{0}$ is a nonempty linear combination
    of at most three vectors from $\vec{C}$
    and with coefficients in $\{1, \ldots, 2 \norm{\vec{C}}^2\}$.
    
    Furthermore, if $\vec{0}$ cannot be expressed like this with fewer than three vectors, the cone of $\vec{C}$ is equal to $\Q^2$.
\end{lemma}
\begin{proof}
   If $\vec{C}$ contains $\vec{0}$, the statement is trivial.
   If $\vec{C}$ contains a vector $\vec{a}$ with a negative coordinate $\vec{a}_i$ as well as a vector $\vec{b}=-\lambda \vec{a}$ for some positive rational $\lambda$, then $\vec{0}$ can be expressed as $\vec{b}_i \vec{a}-\vec{a}_i \vec{b}$ and we are done.
   So now assume that $\vec{C}$ does not contain vectors $\vec{a}$ and $\vec{b}$ like this.

Consider a minimal subset $\vec{C}'\subseteq \vec{C}$ such that $\vec{0}$ can be expressed as a linear combination $\lambda_1 \vec{a}^{(1)}+\cdots +\lambda_{|\vec{C}'|} \vec{a}^{(|\vec{C}'|)}$
with positive rational coefficients $\lambda_i$ of vectors $\vec{a}^{(i)}\in \vec{C}'$.
Assume for contradiction that $|\vec{C}'| > 3$. Then, there must be a closed half-plane containing at least $3$ vectors, say w.l.o.g.~$\vec{a}^{(1)}$, $\vec{a}^{(2)}$, and $\vec{a}^{(3)}$, from $\vec{C}'$. One of these three vectors can be expressed as a non-negative linear combination of the other two.
Without loss of generality assume $\vec{a}^{(1)}= c_1 \vec{a}^{(2)}+c_2 \vec{a}^{(3)}$ with $c_1,c_2 \ge 0$. But then we can write
\[
\vec{0}=\lambda_1 ( c_1 \vec{a}^{(2)}+c_2 \vec{a}^{(3)}) +\lambda_2 \vec{a}^{(2)}+\lambda_3 \vec{a}^{(3)} +\cdots +\lambda_{|\vec{C}'|} \vec{a}^{(|\vec{C}'| )}
\]
and express $\vec{0}$ as a linear combination with positive coefficients of only $|\vec{C}'|-1$ vectors contradicting the minimality of $\vec{C}'$.

Therefore we can choose three vectors $\vec{a},\vec{b},\vec{c}\in\vec{C}$
    such that there are strictly positive $x_1,x_2,x_3$ and
    $x_1\vec{a}+x_2\vec{b}+x_3\vec{c}=\vec{0}$.
    
    The equation has infinitely many solution since we can scale the coefficients. However, if we set $x_3$ to be, say, $|\vec{b}_1\vec{a}_2-\vec{a}_1\vec{b}_2|$ the solution becomes unique (since $\vec{a}$ and $\vec{b}$ are linearly independent) and it can be easily checked that the solution obtained by Cramer's rule is $x_1= |\vec{c}_1\vec{b}_2-\vec{b}_1\vec{c}_2|$ and $x_2 = |\vec{a}_1\vec{c}_2-\vec{c}_1\vec{a}_2|$.

For the second statement of the lemma observe that we can express $-\vec{a}$ and $-\vec{b}$ as linear combinations of $\vec{a}$, $\vec{b}$, and $\vec{c}$ with positive rationals. For example, $-\vec{a}= (x_2\vec{b}+x_3\vec{c})/x_1$. Since $\vec{a}$ and $\vec{b}$ are linearly independent, any vector in $\Q^2$ can be expressed as a linear combination of $\vec{a}$ and $\vec{b}$ using rational coefficients. Combined with the fact that we can express $-\vec{a}$ and $-\vec{b}$ the claim follows.
    \end{proof}

The next two lemmas apply to the other case,
i.e.\ when the cone does not contain the zero vector:
firstly, such cones are determined by pairs of 
outermost vectors in their spanning sets; 
and secondly, they are contained in 
open halfplanes determined by small vectors.

\newcommand{\lvec}[1]{\vec{#1}_\circlearrowleft}
\newcommand{\rvec}[1]{\vec{#1}_\circlearrowright}
Let us write
$\rvec{v} \eqdef \tuple{ \vec{v}_2,-\vec{v}_1}$ and
$\lvec{v} \eqdef \tuple{-\vec{v}_2,\vec{v}_1}$
for the vector $\vec{v}\in\Z^2$ rotated $90^\circ$ clockwise and anticlockwise,
respectively.

\begin{lemma}
    \label{l:outermost}
    If the cone of $\emptyset \neq \vec{C}\subseteq\Z^2$
    does not contain $\vec{0}$, then
    there are two vectors $\vec{a},\vec{b}\in\vec{C}$ such that
    $\{\vec{a,b}\}$ spans the same cone as $\vec{C}$, and for
            all $\vec{c}$ in the cone of $\vec{C}$,
        $\lvec{a}\cdot\vec{c}\ge 0$ and $\rvec{b}\cdot\vec{c}\ge 0$.
\end{lemma}
\begin{proof}
Consider a subset $\vec{C}'\subseteq\vec{C}$ of minimum size that spans the same cone as $\vec{C}$.
Assume for contradiction that $|\vec{C}'| > 2$. Then the set contains three vectors $\vec{x}$, $\vec{y}$, and $\vec{z}$
and because these vectors must be linearly dependent we have $\vec{z} = \lambda_1 \vec{x}+\lambda_2 \vec{y}$ for some rationals $\lambda_1$ and $\lambda_2$.
We can assume, without loss of generality, that $\lambda_1$ and $\lambda_2$ do not have different signs (otherwise we can appropriately rename $\vec{x}$, $\vec{y}$, and $\vec{z}$).
If $\lambda_1$ and $\lambda_2$ are non-negative, $\vec{C}'\setminus \{\vec{z}\}$ still spans the same cone as $\vec{C}'$ since in any positive combination, $\vec{z}$ can be replaced by $\lambda_1 \vec{x}+\lambda_2 \vec{y}$. If however, $\lambda_1$ and $\lambda_2$ are non-positive, the cone spanned by $\vec{C}'$ contains $\vec{0}$ since $\vec{0}= \vec{z} -  \lambda_1 \vec{x} -\lambda_2 \vec{y}$. In both cases we get a contradiction to our assumptions.

So there must indeed be two vectors $\vec{a,b}\in\vec{C}$, not necessarily different,
that span the same cone as $\vec{C}$.
Observe that
$\vec{b}\cdot \lvec{a} < 0
\iff
\vec{b}\cdot \rvec{a}
> 0$ because $\rvec{a} = -\lvec{a}$.
Further observe that $\vec{b}\cdot \rvec{a}=\lvec{b}\cdot \vec{a}$.
Therefore, either
$\vec{b}\cdot \lvec{a} \ge 0$ or $\vec{b}\cdot \rvec{a} = \lvec{b}\cdot \vec{a}  \ge 0$ holds.
We assume w.l.o.g.~that $\vec{b}\cdot \lvec{a} \ge 0$, since otherwise we can swap the names of  $\vec{a}$ and $\vec{b}$.

Pick any $\vec{c}\in\vec{C}$.
Since the cone of $\{\vec{a,b}\}$ contains $\vec{c}$, there exist $x,y\ge 0$ such
that
$\vec{c}=x\vec{a}+y\vec{b}$ and therefore
\begin{align*}
    \vec{c}\cdot\lvec{a} = (x\vec{a}+y\vec{b})\cdot\lvec{a}
    &= x\vec{a}\cdot\lvec{a} + y\vec{b}\cdot\lvec{a} \ge 0,
\end{align*}
because $\vec{a}x\cdot\lvec{a}=0$. Analogously,
using $y\vec{b}\cdot\rvec{b}=0$,
we get
\begin{align*}
    \vec{c}\cdot\rvec{b} = (x\vec{a}+y\vec{b})\cdot\rvec{b}
    &= x\vec{a}\cdot\rvec{b} + y\vec{b}\cdot\rvec{b} \ge 0. \qedhere
\end{align*}
\end{proof}

\begin{lemma}
    \label{l:halfplane}
    If the cone of $\emptyset \neq \vec{C}\subseteqfin\Z^2$
    does not contain $\vec{0}$, then
    there exists a vector $\vec{p} \in \mathbb{Z}^2$ such that
$\norm{\vec{p}} \leq 2 \norm{\vec{C}}$ and
$\vec{p} \cdot \vec{c} > 0$ for all $\vec{c} \in \vec{C}$.
\end{lemma}
\begin{proof}
According to \aref{Lemma}{l:outermost} we have vectors $\vec{a}, \vec{b} \in \vec{C}$ such that
$\{\vec{a,b}\}$ spans the same cone as $\vec{C}$, and for
            all $\vec{c}\in\vec{C}$,
        $\lvec{a}\cdot\vec{c}\ge 0$ and $\rvec{b}\cdot\vec{c}\ge 0$.
        
If $\vec{a}$ alone already spans the same cone as $\vec{C}$, we can choose $\vec{p}=\vec{a}$ and are done.
Otherwise, $\vec{a}$ and $\vec{b}$ are linearly independent and we choose $\vec{p}=\lvec{a}+\rvec{b}$. Clearly $\norm{\vec{p}} \leq 2 \norm{\vec{C}}$.
For any $\vec{c}\in\vec{C}$, $\vec{p}\cdot \vec{c} =(\lvec{a}+\rvec{b})\cdot \vec{c}$. Since $\vec{a}$ and $\vec{b}$ are linearly independent,
$\lvec{a}\cdot \vec{c} \neq 0$ or $\rvec{b}\cdot \vec{c} \neq 0$ and therefore $(\lvec{a}+\rvec{b})\cdot \vec{c} > 0$.
\end{proof}

Our last lemma dealing with cones gives some additional properties for the structure of the cones when it is known that the cone does not contain
some vector. For simplicity, and because it is all we will need later, we focus on the case that $\tuple{0,1}$ is not contained in the cone.

\begin{lemma}
    \label{l:vectorexcludingcones}
    
    Let $\emptyset \neq \vec{C}\subseteqfin\Z^2$ be a set not containing $\vec{0}$.
    If the cone of $\vec{C}$
    does not contain $\tuple{0,1}$, then
    there is a vector $\vec{p} \in \mathbb{Z}^2$ such that
    \begin{itemize}
        \item $\norm{\vec{p}} \leq \norm{\vec{C}}$,
        \item $\vec{p} \cdot \tuple{0,1}<0$,
        \item $\vec{p} \cdot \vec{c} \ge 0$ for all $\vec{c} \in \vec{C}$, and
        \item if $\vec{p}_1 < 0$, then $\rvec{p} \in \vec{C}$.
    \end{itemize}
\end{lemma}
\begin{proof}
We distinguish two basic cases based on whether the cone of $\vec{C}$ contains $\vec{0}$ or not.
First suppose the cone of $\vec{C}$ does not contain $\vec{0}$.
Then, by \aref{Lemma}{l:outermost}, there are vectors
$\vec{a},\vec{b}\in\vec{C}$ such that
    $\{\vec{a,b}\}$ spans the same cone as $\vec{C}$, and for
            all $\vec{c}\in\vec{C}$,
$\lvec{a}\cdot\vec{c}\ge 0$ and $\rvec{b}\cdot\vec{c}\ge 0$. Note that, in particular, we can plug in $\vec{b}$ for $\vec{c}$
and then must have $\lvec{a}\cdot\vec{b}\ge 0$ which implies
$\vec{a}_1\vec{b}_2 - \vec{b}_1\vec{a}_2 \ge 0$.
        
There are three candidates for the choice of $\vec{p}$:
$\lvec{a}$, $\tuple{0,-1}$, and $\rvec{b}$.
Suppose $\vec{p}=\lvec{a}$ does not satisfy all conditions of
the lemma. Then we must have $\lvec{a} \cdot \tuple{0,1} \ge 0$
and therefore $\vec{a}_1 \ge 0$.
Assume further that $\tuple{0,-1}$ also does not satisfy all
conditions of the lemma.
Then there must be a vector $\vec{c}'\in\vec{C}$ with $\vec{c}'_2
> 0$.

We now show that if neither $\lvec{a}$ nor $\tuple{0,-1}$ can be
used for $\vec{p}$, $\rvec{b}$ can.
Assume for contradiction that $\rvec{b} \cdot \tuple{0,1} \ge
0$.
But then $\vec{b}_1 \le 0$ and we can express
$\tuple{0,\vec{a}_1\vec{b}_2-\vec{b}_1\vec{a}_2} =
\vec{a}_1\vec{b}-\vec{b}_1\vec{a}$ as a positive combination
of $\vec{a}$ and $\vec{b}$. Since the cone of $\vec{C}$ does not
contain $\vec{0}$ or $\tuple{0,1}$ we must have
$\vec{a}_1\vec{b}_2-\vec{b}_1\vec{a}_2<0$ which, as we argued
above, cannot be the case.
Here we used the assumption that we do not have $\vec{a}_1=\vec{b}_1=0$. If that
were the case, either $\tuple{0,1}$ would be in the cone (if $\vec{a}_2 > 0$ or $\vec{b}_2 > 0$)
or $\vec{c}'$ could not be in the cone spanned by $\vec{a}$ and $\vec{b}$, which is a contradiction.
                
We conclude that $\rvec{b} \cdot \tuple{0,1} < 0$ and thus
$\vec{b}_1 > 0$. To finish the proof we have to argue that
$\vec{p}=\rvec{b}$ is a valid choice and we do this by showing
that $\vec{p_1}=\vec{b}_2 \ge 0$.
Assume for contradiction that $\vec{b}_2 < 0$. Since $\vec{a}_1
\ge 0$ and $\vec{b}_1\vec{a}_2 \le \vec{a}_1\vec{b}_2 \le 0$, we
can conclude that $\vec{a}_2 \le 0$. However, if $\vec{b}_2 < 0$
and $\vec{a}_2 \le 0$, $\vec{c}'$ cannot be in the cone spanned
by $\vec{a}$ and $\vec{b}$, which is a contradiction.
 
 We now move to the second case in which we assume that the cone of $\vec{C}$ does contain $\vec{0}$.
 Then there are two vectors $\vec{a}, \vec{b} \in \vec{C}$ such that $\vec{a}+\lambda\vec{b}=\vec{0}$ for some positive rational $\lambda$.
This is because if $\vec{0}$ could only be expressed with three or more vectors, according to \aref{Lemma}{l:small.zero}, $\tuple{0,1}$ would also be in the cone of $\vec{C}$. 

 Clearly, either $\vec{a}_1 \le 0$ or $\vec{b}_1 \le 0$. Without loss of generality let $\vec{a}_1 \le 0$.
 We choose $\vec{p} = \lvec{a}$.
 
 The only condition of the lemma not trivially met is that $\vec{p}\cdot \vec{c}\ge 0$ for all $\vec{c}\in\vec{C}$. Assume that there is a $\vec{c}\in\vec{C}$
 such that $\lvec{a}\cdot \vec{c}< 0$. Then $\vec{c}_1\vec{a}_2-\vec{a}_1\vec{c}_2 > 0$. If $\vec{c}_1 \ge 0$, $\tuple{0,1}$ would be in the cone of $\vec{C}$ since it can be expressed as
$\vec{c}_1 \cdot \vec{a} - \vec{a}_1 \cdot \vec{c}/(\vec{c}_1\vec{a}_2-\vec{a}_1\vec{c}_2)$. Otherwise $\tuple{0,1}$ would also be in the cone of $\vec{C}$ since it can be expressed as 
 $\vec{b}_1 \cdot \vec{c} - \vec{c}_1 \cdot \vec{b}/(\vec{b}_1\vec{c}_2-\vec{c}_1\vec{b}_2)$. Either way, we have a contradiction. 
\end{proof}

Moving from rational cones to paths of \SLPS s,
our remaining two lemmas pin down some relatively basic properties
of \SLPS\ paths in which some cycles are repeated `many' times:
firstly, if all those cycles are contained in a halfplane,
then the effect of the path must point roughly in the same direction
(we have a strict and a non-strict version here);
secondly, if the path when started at a point remains sufficiently far 
from both axes (i.e.\ respects a sufficiently wide margin),
then it can be shortened admissibly by a range of multiples of any small vector
that is in the cone spanned by the `often' repeated cycles.

For a path $\pi$ of a $2$-\SLPS\ 
$\Lambda = \alpha_0
    \beta_1^*
    \alpha_1
    \beta_2^*
    \dots
    \beta_K^*
    \alpha_K$
and a bound $B \in \N$, let 
$$\Cycles{B}{\Lambda}{\pi} \,\subseteq\, \Z^2$$
be the set of all cycles of $\Lambda$ 
that are repeated in $\pi$ at least $B$ times.

\begin{lemma}
\label{l:drift}
Suppose $\pi$ is a path of a $2$-\SLPS\ $\Lambda$ with $K$ cycles,
$B \in \N$ and $\vec{p} \in \Z^2$.
\begin{description}
\item[(i)] If $\vec{p} \cdot \vec{a} > 0$ for all $\vec{a} \in \Cycles{B}{\Lambda}{ \pi}$, then
\[\vec{p} \cdot \effect{\pi} \,\geq\,
  |\pi| - (K B + 1) (2 \norm{\Lambda} \, \norm{\vec{p}} + 1).\]
\item[(ii)]
If $\vec{p} \cdot \vec{a} \geq 0$ for all 
$\vec{a} \in \Cycles{B}{\Lambda}{ \pi}$, then
\[\vec{p} \cdot \effect{\pi} \,\geq\,
  - (K B + 1) (2 \norm{\Lambda} \, \norm{\vec{p}}).\]
\end{description}
\end{lemma}
\begin{proof}
    The effect of $\pi$
    can be decomposed as $\effect{\pi}=\vec{v} + \vec{b}$,
    where $\vec{v}$ is the combined effect of those cycles occurring
    at least $B$ times and $\vec{b}$ is the rest.
    Hence $\vec{v}$ is a linear combination
    $\vec{v}=\sum_{i=1}^\ell\vec{a}^{(i)}$, where
    $\vec{a}^{(i)}\in\Cycles{B}{\Lambda}{\pi}$
    and
    $\vec{b}$ is the effect of a path
    of length $\len{\pi}-\ell \le KB +1$.
    We can therefore estimate
   \begin{equation}
    \vec{p}\cdot\vec{b}\ge
    -2(KB+1)\norm{\Lambda}\norm{\vec{p}}.
   \label{eq:drift1}
   \end{equation} 
   If $\vec{p}\cdot\vec{a}>0$ for all
   $\vec{a}\in\Cycles{B}{\Lambda}{\pi}$ then
   \begin{equation}
   \vec{p}\cdot\vec{v}
   =
   \sum_{i=1}^\ell \vec{p}\cdot\vec{a}^{(i)}
   ~\ge~\ell~\ge~
   \len{\pi}-(KB+1).
   \label{eq:drift2}
   \end{equation}
   The first claim therefore follows by \aref{Equations}{eq:drift1} and \ref{eq:drift2}
   and by the fact that $\vec{p}\cdot\effect{\pi}=\vec{p}\cdot\vec{v}+\vec{p}\cdot\vec{b}$.

   For the second claim, just observe that
   if $\vec{p}\cdot\vec{a}\ge0$ for all
   $\vec{a}\in\Cycles{B}{\Lambda}{\pi}$,
   then
   $\vec{p}\cdot\vec{v} = \sum_{i=1}^l \vec{p}\cdot\vec{a}^{(i)}\ge 0$.
\end{proof}

Let us call a path $\pi'$ a \emph{shortening of path $\pi$ by vector $\vec{e}$}
when $\pi'$ is a proper subword (not necessarily contiguous) of $\pi$
and $\effect{\pi'}=\effect{\pi} - \vec{e}$.

\begin{lemma}
\label{l:cut}
Suppose a path $\pi$ of a $2$-\SLPS\ $\Lambda$,
$N \in \N$, $\vec{c} \in \Z^2$ and $\vec{s} \in \N^2$ satisfy:
\begin{itemize}
\item
$\norm{\Lambda} > 0$ and $\norm{\vec{c}}\le \norm{\Lambda}$,
\item
the cone of $\Cycles{2 \norm{\Lambda}^2 N}{\Lambda}{ \pi}$ contains $\vec{c}$, and
\item 
all points visited by $\pi$ from $\vec{s}$ are in 
$(\N_{{\geq}\, 6 \norm{\Lambda}^3 N})^2$.
\end{itemize}
There exists $\gamma \in \{1, \ldots, 2 \norm{\Lambda}^2\}$ such that,
for all $n \in \{1, \ldots, N\}$,
$\pi$ has a shortening by $n \gamma \vec{c}$
which is admissible from~$\vec{s}$.
\end{lemma}

\begin{proof}
Let $\vec{C}\eqdef\Cycles{2 \norm{\Lambda}^2 N}{\Lambda}{ \pi}$.
We claim that
\[\gamma \vec{c} \,=\,
  \lambda_1 \vec{a}^{(1)} + \cdots + \lambda_j \vec{a}^{(j)}\]
for some $j \in \{1, 2, 3\}$,
$\vec{a}^{(1)}, \ldots, \vec{a}^{(j)} \in \vec{C}$ and
$\gamma, \lambda_1, \ldots, \lambda_j \in \{1, \ldots, 2 \norm{\Lambda}^2\}$.
If $\vec{c} = \vec{0}$, this directly follows
from \aref{Lemma}{l:small.zero}.
Otherwise, reasoning as in the proof of \aref{Lemma}{l:outermost},
there must be two vectors $\vec{a}^{(1)}, \vec{a}^{(2)} \in \vec{C}$
such that the cone spanned by $\{\vec{a}^{(1)}, \vec{a}^{(2)}\}$ 
contains $\vec{c}$ but not $\vec{0}$.  Then the claim follows 
by \aref{Lemma}{l:small.zero} applied to the set
$\{-\vec{c}, \vec{a}^{(1)}, \vec{a}^{(2)}\}$.

Now, we can subtract $n \gamma \vec{c}$ from the effect of $\pi$
by deleting $n \lambda_i$ occurences of the cycle $\vec{a}^{(i)}$ 
for all $i \in \{1, \ldots, j\}$.
Any such shortening $\pi'$ is admissible from $\vec{s}$ because, 
for any point visited by $\pi$, the differences between its coordinates and 
the coordinates of the corresponding point visited by $\pi'$ 
are at most $6 \norm{\Lambda}^3 n$.
\end{proof}

\begin{figure}[t]
  \includegraphics[width=0.5\linewidth]{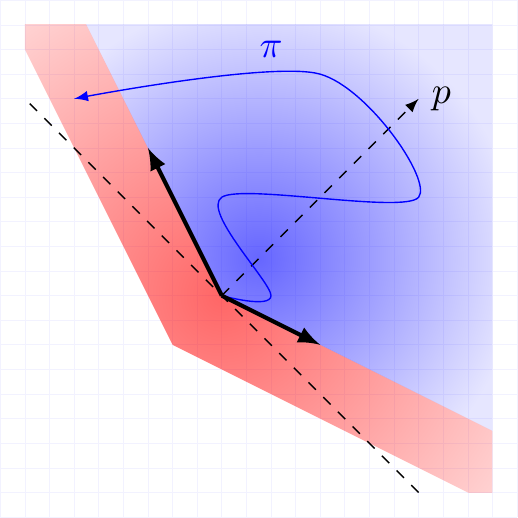}
  \includegraphics[width=0.5\linewidth]{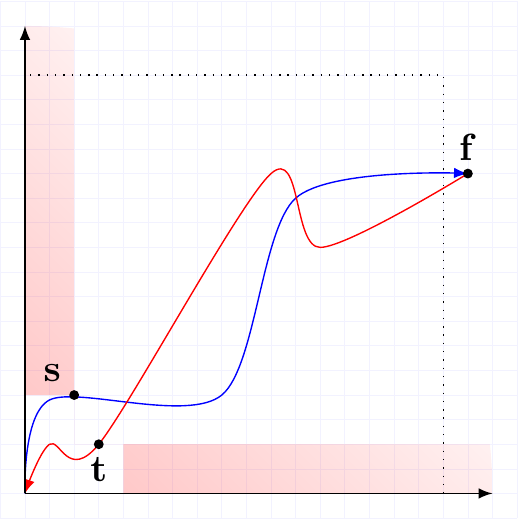}
  \caption{\aref{Lemma}{l:drift} (left):
      The path $\pi$ must remain in the red/blue area.
      In case (ii) the red belt is parallel to the dashed line,
      i.e.~orthogonal to $\vec{p}$.
      \aref{Theorem}{th:far} (right):
      The path from $\vec{s}$ to $\vec{t}$ via a sufficiently large
      point $\vec{f}$ can be shortened.
  }
  \label{fig:pics3}
\end{figure}

\section{Go!}
\label{sec:lps-short}
Here is the bulk of our work.

We present a sequence of theorems that culminates in 
\aref{Theorem}{thm:2SLPS-short}, which establishes that 
if a reachability witness of a $2$-dimensional \textSLPS\ cannot be shortened, 
then it cannot visit points whose norm exceeds a certain polynomial bound 
(in the length and the norm of the \SLPS).

A key step towards the last theorem is \aref{Theorem}{th:far},
where lemmas from the previous section are employed to conclude that
it suffices to prove that shortest reachability witnesses
cannot visit points that are `near' one of the axes 
but further from the other axis than a certain polynomial bound 
(smaller than the one in \aref{Theorem}{thm:2SLPS-short}, see the red margins in
\aref{Figure}{fig:pics3} on the right).

The remainder of our reasoning here is therefore concerned with showing that
shortest reachability witnesses cannot contain points that are,
without loss of generality, within a $y$-axis margin but 
too far from the $x$-axis (more than a polynomial bound).
We accomplish this by proving that, if such a scenario occurs,
then we can focus on a point $\vec{t}$ that is 
within the $y$-axis margin and maximally far from the $x$-axis, 
and find an admissible shortening of the reachability witness 
whose effect on $\vec{t}$ is to decrease its $y$-coordinate by a `small' amount.

Theorems~\ref{th:close.away}--\ref{th:one.visit} provide 
increasingly powerful tools for identifying admissible shortenings 
of paths that in some way climb the $y$-axis.
In the proof of \aref{Theorem}{thm:2SLPS-short},
such shortenings are applied to appropriate 
segments and reversals of segments of reachability witnesses.  
Thus their effects have to be matched
(recall $1$-dimensional hill cutting \citep[cf.\ e.g.][]{VP1975}),
which explains the ranges of possible shortenings in 
Theorems~\ref{th:close.away}--\ref{th:one.visit}.

\begin{figure}[t]
  \includegraphics[width=0.4\linewidth]{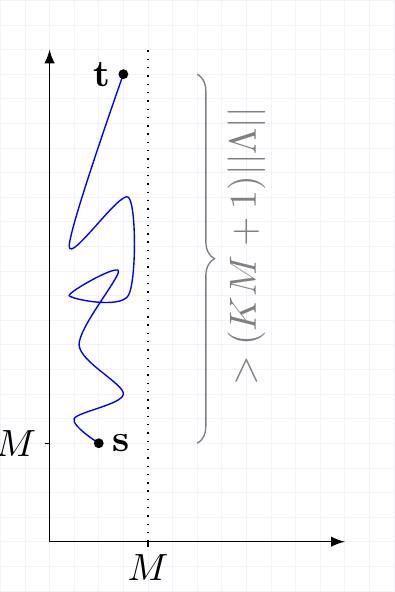}
  \includegraphics[width=0.6\linewidth]{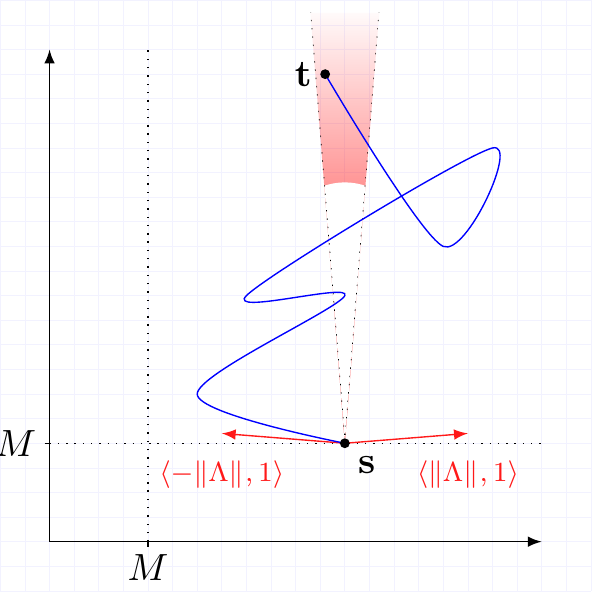}
  \caption{Illustrations of \aref{Theorem}{th:close.away} (left) and
      \aref{Theorem}{th:away.both} (on the right,
      with $M\eqdef 6 \norm{\Lambda}^3 N$).
  }
  \label{fig:pics}
\end{figure}

We begin with handling the case in which a path goes up by a large amount
but only visits points which are close to the $y$-axis and not close to the
$x$-axis, cf.\ \aref{Figure}{fig:pics} on the left.

Recall that, for planar vectors $\vec{v}$, 
we denote their horizontal and vertical components by 
$\vec{v}_1$ and $\vec{v}_2$, respectively.

\begin{theorem}
\label{th:close.away}
Suppose a $2$-\SLPS\ $\Lambda$ with at most $K$ cycles has
a path $\pi$ from point $\vec{s}$ to point $\vec{t}$ such that
for some $M\in\N$
\begin{itemize}
\item
all points visited by $\pi$ from $\vec{s}$ are in
$\mathbb{N}_{{<} M} \times \mathbb{N}_{{\geq} M}$ and
\item
$(\vec{t} - \vec{s})_2 > (K M + 1) ||\Lambda||$.
\end{itemize}
There is $\gamma \in \{1, \ldots, ||\Lambda||\}$ such that,
for all $n \in \{1, \ldots, \lfloor M / \gamma \rfloor\}$,
$\pi$ has a shortening by $\tuple{0, n \gamma}$
which is admissible from~$\vec{s}$.
\end{theorem}

\begin{proof}
There is a cycle $\vec{c}$ in $\pi$ that is repeated at least
$M$ times. Otherwise,
for the effect of $\pi$, $||\vec{t}-\vec{s}|| \le
((K+1)+K\cdot(M-1))\cdot ||\Lambda||=(KM+1)||\Lambda||$, which
contradicts the second assumption of the theorem.
Let $\vec{u}$ and $\vec{v}$ be the points visited right before the first,
and right after the last repetitions of the cycle $\vec{c}$, respectively.
The first coordinate of $\vec{c}$ is 0 since otherwise
$|(\vec{u}-\vec{v})_1|\ge M$, which contradicts the first
assumption of the theorem.
Therefore $\vec{c}=\tuple{0,\gamma}$ for some $\gamma \in \{1, \ldots,
||\Lambda||\}$ and thus, $\pi$ has a shortening by
$\tuple{0, n\gamma}$ for all $n \in \{1, \ldots, \lfloor M / \gamma \rfloor\}$.
This shortening is admissible
since it does not affect the first coordinate of any point
visited, only decreases the second coordinates by at most $
\lfloor M / \gamma \rfloor \cdot \gamma \le M$, and
all visited points have a second coordinate value of at least $M$ prior to the
shortening.
\end{proof}

The following theorem deals with a case in which all points visited on a path
are far from both axes but where the total effect of the path is much bigger in
the second coordinate than the first, cf.\ \aref{Figure}{fig:pics} on the right,
where $M = 6 ||\Lambda||^3 N$.

\begin{theorem}
\label{th:away.both}
Suppose a $2$-\SLPS\ $\Lambda$ with at most $K$ cycles has 
a path $\pi$ from point $\vec{s}$ to point $\vec{t}$ such that
for some $N\in\N$
\begin{itemize}
\item
all points visited by $\pi$ from $\vec{s}$ are in
$(\mathbb{N}_{{\geq}\, 6 ||\Lambda||^3 N})^2$ and
\item 
for all $\lambda \in [-||\Lambda||,||\Lambda||]$,
$\tuple{ \lambda, 1} \cdot (\vec{t} - \vec{s})
 \,>\, (4 K N + 2) ||\Lambda||^4$.
\end{itemize}
$\tuple{0, 1}$ is in the cone of $\Cycles{2 ||\Lambda||^2 N}{\Lambda}{ \pi}$
and there exists $\gamma \in \{1, \ldots, 2 ||\Lambda||^2\}$ such that,
for all $n \in \{1, \ldots, N\}$,
$\pi$ has a shortening by $\tuple{0, n \gamma}$
which is admissible from~$\vec{s}$.
\end{theorem}

\begin{proof}
Note that  $(\vec{t} - \vec{s})_2>0$, since otherwise either $\tuple{1, 1} \cdot (\vec{t} - \vec{s})$ or $\tuple{-1, 1} \cdot (\vec{t} - \vec{s})$ would be non-positive contradicting
the assumption of the theorem.

Let $\vec{C} = \Cycles{2 ||\Lambda||^2 N}{\Lambda}{ \pi}$.
Assume for contradiction that $\tuple{0, 1}$ is not in the cone of $\vec{C}\setminus \{\vec{0}\}$.
Then, due to \aref{Lemma}{l:vectorexcludingcones}, there exists $\vec{p} \in \mathbb{Z}^2$ such that
$||\vec{p}|| \leq ||\Lambda||$,
$\vec{p} \cdot \tuple{0, 1} < 0$, and
$\vec{p} \cdot \vec{a} \geq 0$ for all $\vec{a} \in \vec{C}$.
This implies $\vec{p}_2 < 0$ and therefore
\begin{align*}
-\vec{p} \cdot (\vec{t} - \vec{s}) \ge \tuple{-\vec{p}_1,1} \cdot (\vec{t} - \vec{s}) >  (4 K N + 2) ||\Lambda||^4.
\end{align*}
But, by \aref{Lemma}{l:drift} (ii),
\begin{align*}
\vec{p} \cdot (\vec{t} - \vec{s})
& \geq - (K 2 \norm{\Lambda}^2N + 1) (2 \norm{\Lambda}^2) \\
& \geq - (4 K N + 2) \norm{\Lambda}^4.
\end{align*}

Therefore, $\tuple{0, 1}$ must be in the cone of $\vec{C}$,
and we conclude by \aref{Lemma}{l:cut}.
\end{proof}

\begin{figure}[t]
  \includegraphics[width=0.5\linewidth]{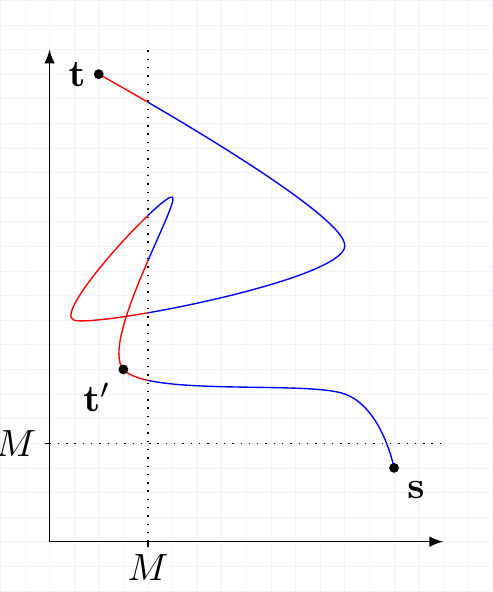}
  \includegraphics[width=0.5\linewidth]{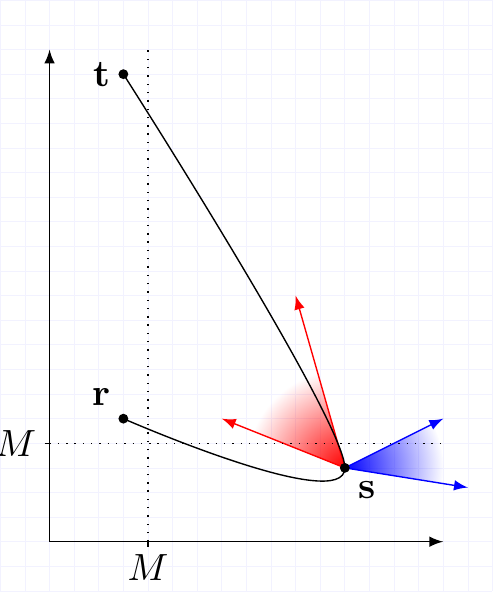}
  \caption{
  \aref{Theorem}{th:away.other} (left):
  Either the cone of cycles in the segment from $\vec{t'}$ to $\vec{t}$ contains $\tuple{0,1}$
  or that from $\vec{s}$ to $\vec{t'}$ contains some short vector in the top-left
  quadrant.
  \aref{Theorem}{th:one.visit} (right):
  The cone of cycles from $\vec{r}$ to $\vec{s}$ (blue), combined with
  the cone of cycles from $\vec{s}$ to $\vec{t}$ (red), contains $\tuple{0,1}$.
  }
  \label{fig:pics2}
\end{figure}
 
In the next theorem we combine the previous two results to handle the case when a path starts close to the $x$-axis, ends close to the $y$-axis but far away from the $x$-axis and does not come close the the $x$-axis anywhere in between,
cf.\ \aref{Figure}{fig:pics2} on the left.

\begin{theorem}
\label{th:away.other}
Suppose $N \in \mathbb{N}$, 
$M \geq 6 \norm{\Lambda}^3 N$, and 
a $2$-\SLPS\ $\Lambda$ with $K > 0$ cycles has
a path $\pi$ from point $\vec{s}$ to point $\vec{t}$ such that
\begin{itemize}
\item
$\vec{s}_1 \geq 0$,
$\vec{s}_2 < M$,
\item
$\vec{t}_1 < M$,
$\vec{t}_2 \,\geq\, 12 (K+1)(M+1) ||\Lambda||^4$ and
\item
all points visited by $\pi$ after $\vec{s}$ are in
$\mathbb{N} \times \mathbb{N}_{{\geq} M}$.
\end{itemize}
Let $\pi'$ be the shortest nonempty prefix of $\pi$ 
whose target point $\vec{t'}$ satisfies $\vec{t'}_1 < M$.
Provided $|\pi'| \geq 2$, 
let $\pi^\dag$ be $\pi'$ without its first and last vectors,
let $\Lambda^\dag$ be an \SLPS\ one of whose paths is $\pi^\dag$ 
and whose length and norm are at most those of $\Lambda$, and
let $\vec{C} = \Cycles{2 ||\Lambda||^2 N}{\Lambda^\dag}{\pi^\dag}$.
\begin{description}
\item[(i)]
If either $(\vec{t} - \vec{t'})_2 \,>\, 6(K+1)(M+1) ||\Lambda||^4$
or $\tuple{0, 1}$ is in the cone of $\vec{C}$,
then there exists $\gamma \in \{1, \ldots, 2 ||\Lambda||^2\}$ such that,
for all $n \in \{1, \ldots, N\}$,
$\pi$ has a shortening by $\tuple{0, n \gamma}$
which is admissible from~$\vec{s}$.
\item[(ii)]
Otherwise, there exists 
$\vec{v} \,\in\, \vec{C} \cap (\mathbb{Z}_{{<} 0} \times \mathbb{Z}_{{>} 0})$ 
such that
\[
  \lvec{v} \cdot \tuple{\vec{s}_1, -\vec{t}_2}\,<\,7(K+2)(M+1) ||\Lambda||^5.
\]
\end{description}
\end{theorem}

\begin{proof}
If $(\vec{t} - \vec{t'})_2 \,>\, 6(K+1)(M+1) ||\Lambda||^4$,
let $\pi''$ be the rest of $\pi$ after $\pi'$, i.e., the segment of $\pi$ that starts at $\vec{t'}$ and ends at $\vec{t}$.
Then partition $\pi''$ into
segments that visit only points in $\mathbb{N}_{{<} M} \times \mathbb{N}$ and 
segments for which all intermediate points are outside that set. Call these segments $y$-axis-close and $y$-axis-far, respectively.
In the following we argue that either \aref{Theorem}{th:close.away} applies to one of the former segments, 
or \aref{Theorem}{th:away.both} applies to one of the latter segments.

Let $\ell$ be the total number of segments and, for $i\in [1,\ell-1]$, let $\vec{a}^{(i)}$ be the endpoint of the $i$-th segment and the start point of the ($i+1$)-th segment. Note that a path from an \SLPS\ with at most $K$ cycles will be split into at most $2(K+1)$ segments and therefore $\ell \le 2(K+1)$.
For convenience, define $\vec{a}^{(0)}$ to be $\vec{t'}$ and $\vec{a}^{(\ell)}$ to be~$\vec{t}$.

Each segment corresponds to a \SLPS\ that is a fragment of the original \SLPS. Let the \SLPS\ fragment of the $i$-th segment contain $K_i$ cycles. Note that each of the cycles in the original \SLPS\ can only be part of two
 different segments. Therefore, $\sum K_i \le 2K$. 

\noindent
Since
\begin{align*}
  & \sum_{i=1}^{\ell} (\vec{a}^{(i)}-\vec{a}^{(i-1)})_2 
    = (\vec{t} - \vec{t'})_2 \\
  & >6(K+1)(M+1)||\Lambda||^4 \\
  & > 2KM||\Lambda|| + 2(K+1)(M+1)||\Lambda|| + 4(K+1)||\Lambda||^4,
\end{align*}
there must be a segment $i$,
going from $\vec{a}^{(i-1)}$ to $\vec{a}^{(i)}$,
for which
 $$(\vec{a}^{(i)}-\vec{a}^{(i-1)})_2 > (K_i M+M+1)||\Lambda|| + 2||\Lambda||^4.$$
If this segment $i$ is $y$-axis-close, we observe that
$(\vec{a}^{(i)}-\vec{a}^{(i-1)})_2 > (K_i M+1)||\Lambda||$
and therefore \aref{Theorem}{th:close.away} applies to it.

If this segment $i$ is $y$-axis-far then
\begin{align*}
(\vec{a}^{(i)}-\vec{a}^{(i-1)})_2 & > (K_i 6N||\Lambda||^3+M+1)||\Lambda|| + 2||\Lambda||^4 \\ & > (K_i 4N+2)||\Lambda||^4+2||\Lambda||+M||\Lambda||,
\end{align*}
since $M\ge 6||\Lambda||^3 N$.

Now consider the point $\vec{a}^{(i-1)'}$ visited right after $\vec{a}^{(i-1)}$ and
the point $\vec{a}^{(i)'}$ visited right before $\vec{a}^{(i)}$ and consider the
path between $\vec{a}^{(i-1)}$ and $\vec{a}^{(i)}$ without the first and last
vector.
Note that $\vec{a}^{(i-1)'}_1, \vec{a}^{(i)'}_1 \in [M,  M + ||\Lambda||)$ and hence
$|(\vec{a}^{(i)'}-\vec{a}^{(i-1)'})_1| \le ||\Lambda|| < M$. Therefore, we have 
 $\tuple{ \lambda, 1}
\cdot (\vec{a}^{(i)'}-\vec{a}^{(i-1)'}) >  (K_i 4N+2)||\Lambda||^4$ for all $\lambda
\in [-||\Lambda||,||\Lambda||]$ and hence \aref{Theorem}{th:away.both} applies
to this subpath, going from $\vec{a}^{(i-1)'}$ to~$\vec{a}^{(i)'}$.

Note that the section of $\pi$ going from $\vec{s}$ to $\vec{a}^{(i-1)}$ (or
$\vec{a}^{(i-1)'}$, respectively) is still admissible after the shortening
carried out through \aref{Theorem}{th:close.away} or
\aref{Theorem}{th:away.both}. The shortened segment $i$ is also admissible due
to these theorems. The section of $\pi$ that started at $\vec{a}^{(i)}$ prior to
the shortening is also admissible since the first coordinate of the
corresponding points is not changed and the second coordinate is decreased by at
most $N2\norm{\Lambda}^2 < M$. Moreover, the second coordinate of all the points prior to the shortening was at least~$M$.

In the remainder of the proof, assume
\[ (\vec{t} - \vec{t'})_2 \,\le\, 6(K+1)(M+1) ||\Lambda||^4 \]
and consequently 
\[  \vec{t'}_2 \,\ge\, 6(K+1)(M+1) ||\Lambda||^4, \]
since $\vec{t}_2 \ge 12(K+1)(M+1) ||\Lambda||^4$.

Then $|\pi'| \geq 2$, so $\pi^\dag$, $\Lambda^\dag$ and $\vec{C}$ are well defined.
Let $\vec{s^\dag}$ be the first point visited by $\pi'$ after $\vec{s}$, and
let $\vec{t^\dag}$ be the target point of $\pi^\dag$ from~$\vec{s^\dag}$.
Observe, that
$\vec{s^\dag}_1 \ge 0$,~
$\vec{s^\dag}_2 < M+||\Lambda||$,~
$\vec{t^\dag}_1 < M + ||\Lambda||$, and
$\vec{t^\dag}_2 \ge  6(K+1)(M+1) ||\Lambda||^4-||\Lambda||$.

If $\tuple{0, 1}$ is in the cone of $\vec{C}$,
we are done by \aref{Lemma}{l:cut} applied to $\pi^\dag$ from~$\vec{s^\dag}$,
which visits only points in~$(\mathbb{N}_{{\geq} M})^2$.
Note that all points of $\pi$ after $\vec{s}$ have a second coordinate of at least $M$. Therefore, the shortening due to \aref{Lemma}{l:cut} can also be applied to $\pi$ and result in an admissible path from~$\vec{s}$.

If $\tuple{0, 1}$ is not in the cone of $\vec{C}\setminus\{\vec{0}\}$
then \aref{Lemma}{l:vectorexcludingcones} provides
a vector $\vec{v} \in \Z^2$ such that
$||\vec{v}|| \leq ||\Lambda||$,
$\lvec{v} \cdot \tuple{0, 1} < 0$,
$\lvec{v} \cdot \vec{a} \geq 0$ for all $\vec{a} \in \vec{C}$, and
such that $\vec{v}_2 > 0$ implies $\vec{v} \in \vec{C}$.
Hence, $\vec{v}_1 < 0$ and \aref{Lemma}{l:drift} (ii) gives us
\begin{equation}
    \label{eq:thm8e}
\begin{aligned}
\lvec{v} \cdot (\vec{t^\dag} - \vec{s^\dag})
& \geq - (2 K N ||\Lambda||^2 + 1) (2 ||\Lambda||^2) \\
& \geq - 2 K (2 N + 1) ||\Lambda||^4.
\end{aligned}
\end{equation}
But then $\vec{v}_2 > 0$, since the contrary would contradict
\aref{Equation}{eq:thm8e}:
\begin{align*}
\lvec{v} \cdot (\vec{t^\dag} - \vec{s^\dag})
& <    
-\vec{v}_2 (M+||\Lambda||)\\
&\quad+\vec{v}_1 ( 6(K+1)(M+1) ||\Lambda||^4 -M-2||\Lambda||) \\
& \le  ||\Lambda|| (M+||\Lambda||)\\
&\quad- ( 6(K+1)(M+1) ||\Lambda||^4 -M-2||\Lambda||) \\
& \le  (-6K(M+1) - 6(M+1) +3+2M)||\Lambda||^4 \\
& \le  -6K(M+1)||\Lambda||^4 \\
& \le  - 2 K (2 N + 1) ||\Lambda||^4,
\end{align*}
where the last step follows since $M\ge 6||\Lambda||^3 N$.
Hence $\vec{v} \in \vec{C}$.

Recalling 
$(\vec{t^\dag} - \vec{s^\dag})_1 \geq - \vec{s^\dag}_1  \geq -\vec{s}_1 - ||\Lambda||$ 
and 
\begin{align*}
(\vec{t^\dag} - \vec{s^\dag})_2
&\ge (\vec{t'} - \vec{s})_2 - 2||\Lambda||\\
&\ge \vec{t}_2-(\vec{t}-\vec{t'})_2 - M - 2||\Lambda||\\
&\ge \vec{t}_2-6(K+1)(M+1)||\Lambda||^4 - M - 2||\Lambda||\\
&\ge \vec{t}_2-6(K+2)(M+1)||\Lambda||^4,
\end{align*}
we then conclude that
\begin{align*}
    &\lvec{v} \cdot \tuple{\vec{s}_1, -\vec{t}_2} \\
&\leq \tuple{-\vec{v}_2, \vec{v}_1}
     \cdot (\vec{s^\dag} - \vec{t^\dag}
            - \tuple{||\Lambda||, 6(K+2) (M+1) ||\Lambda||^4}) \\
&<    2 K (2 N + 1) ||\Lambda||^4
     + \vec{v}_2 ||\Lambda||
     - \vec{v}_1 6(K+2) (M+1) ||\Lambda||^4\\
&\leq 7(K+2)(M+1) ||\Lambda||^5.
\qedhere
\end{align*}
\end{proof}

Roughly speaking, our final case deals with a scenario in which the path
consists of two parts. The first part goes from close to the $y$-axis to close
to the $x$-axis without being close to the $x$-axis anywhere in between. In the
second part it goes back, from close to the $x$-axis to close to the $y$-axis without being close to the $y$-axis anywhere in between. 
See \aref{Figure}{fig:pics2} on the right.

\begin{theorem}
\label{th:one.visit}
Suppose $N \in \mathbb{N}$, 
$M \geq 8 \norm{\Lambda}^4 N$, and 
a $2$-\SLPS\ $\Lambda$ with $K > 0$ cycles has
a path $\rho \pi$ consisting of one segment $\rho$ from $\vec{r}$ to $\vec{s}$ and a second segment $\pi$ from $\vec{s}$ to $\vec{t}$ such that
\begin{itemize}
\item
$\vec{r}_1 < M$,
$\vec{r}_2 \geq 0$,
\item
$\vec{s}_1 \geq 0$,
$\vec{s}_2 < M$,
\item
$\vec{t}_1 < M$,
$\vec{t}_2 \,\geq\,  19(K+2) (M+1) ||\Lambda||^6$,
$\vec{t}_2 \geq \vec{r}_2$,
\item
all points visited by $\rho$ after $\vec{r}$ are in
$\mathbb{N}_{{\geq} M} \times \mathbb{N}$ and
\item
all points visited by $\pi$ after $\vec{s}$ are in
$\mathbb{N} \times \mathbb{N}_{{\geq} M}$.
\end{itemize}
There exists $\gamma \in \{0, \ldots, 2 ||\Lambda||^3\}$ such that,
for all $n \in \{1, \ldots, N\}$,
$\rho \pi$ has a shortening by $\tuple{0, n \gamma}$
which is admissible from~$\vec{r}$.
\end{theorem}

\begin{proof}
If case~(i) of \aref{Theorem}{th:away.other} 
applies to $\pi$ from $\vec{s}$ 
then we are done immediately, 
so assume case~(ii) applies to it.

Hence, for some cycle 
$\vec{v} \,\in\, \mathbb{Z}_{{<} 0} \times \mathbb{Z}_{{>} 0}$
which occurs in $\pi$ at least $2 ||\Lambda||^2 N$ times,
we have
\[\lvec{v} \cdot \tuple{\vec{s}_1, -\vec{t}_2}  \,<\, 7(K+2)(M+1) ||\Lambda||^5.\]
This also implies $\vec{s}_1 \ge 12(K+1)(M+1)||\Lambda||^4$, since otherwise
\begin{align*} 
\lvec{v} \cdot \tuple{\vec{s}_1, -\vec{t}_2}
&> -\vec{v}_2 12(K+1)(M+1)||\Lambda||^4 -  \vec{v}_1  \vec{t}_2\\
&\ge -12(K+1)(M+1)||\Lambda||^5 + \vec{t}_2\\
&\ge 7(K+2) (M+1) ||\Lambda||^5.
\end{align*}
Consequently,
\aref{Theorem}{th:away.other} 
with $N \norm{\Lambda}$ for $N$ and with the axes swapped
applies to $\rho$ from $\vec{r}$.

Suppose that case~(ii) of \aref{Theorem}{th:away.other} holds.
That is, for some cycle 
$\vec{w} \,\in\, \mathbb{Z}_{{>} 0} \times \mathbb{Z}_{{<} 0}$
which occurs in $\rho$ at least $2 ||\Lambda||^3 N$ times,
we have
\[\tuple{-\vec{w}_1, \vec{w}_2} \cdot \tuple{\vec{r}_2, -\vec{s}_1}
  \,<\, 7(K+2) (M+1) ||\Lambda||^5.\]
We will reduce the occurrence of cycle $\vec{w}$ in $\rho$ by  $-\vec{v}_1\cdot n$ resulting in a shortening by $-\vec{v}_1\cdot n\cdot\vec{w}$.

If case~(i) of \aref{Theorem}{th:away.other} 
with $N ||\Lambda||$ for $N$ and with the axes swapped
applies to $\rho$ from $\vec{r}$, there is a value $\gamma'\in\{1,\ldots,2||\Lambda||^2\}$ such that we can
shorten $\rho$  by $-\vec{v}_1\cdot n\cdot\tuple{\gamma',0}$. For convenience, we define $\vec{w} \eqdef \tuple{\gamma',0}$ in this case.

Either way, the resulting shortened version of $\rho$ is admissible from $\vec{r}$. In both cases, the second coordinate of points cannot decrease due to the shortening (note that $\vec{w}_2\le0$). The first coordinate may decrease but by at most $||\Lambda||\cdot N ||\Lambda||  \cdot 2||\Lambda||^2 = 2N||\Lambda||^4 < M$. Therefore, the shortened version of $\rho$ is still admissible since, prior to the shortening, all points visited by $\rho$ after $\vec{r}$ have a first coordinate of at least~$M$.

Note that, while $\rho$ is still admissible after the shortening, $\rho\pi$ may not be admissible anymore. Therefore, we also need to shorten $\pi$ appropriately to counter the effect that the shortening of $\rho$ may have had on the first coordinate.
We shorten $\pi$ by reducing the number of occurrences of cycle $\vec{v}$ by $\vec{w}_1\cdot n$. We now argue that such a shortened version of $\pi$ is admissible from $\vec{s}+\vec{v}_1\cdot n\cdot\vec{w}$.

Following \aref{Theorem}{th:away.other}, $\pi$ consists of two parts: a prefix of $\pi$, $\pi'$ for which all intermediate points lie in $(\mathbb{N}_{{\geq} M})^2$, and the remaining path after $\pi'$. Note that the cycle $\vec{v}$ is part of the path $\pi'$. Therefore the target point of $\pi'$ as well as all points on the second part of $\pi$ experience an increase of their first coordinate by $-\vec{v}_1\cdot n\cdot\vec{w}_1$. Hence, after the shortening, all points on $\pi$ starting at $\vec{s}$ have a first coordinate of at least $\min\{M,  -\vec{v}_1\cdot n\cdot\vec{w}_1\}=-\vec{v}_1\cdot n\cdot\vec{w}_1$. Reducing the repetitions of the cycle $\vec{v}$ by $\vec{w}_1\cdot n$ 
can decrease the second coordinates of points on the path by no more than $\vec{w}_1\cdot n\cdot\vec{v}_2 \le 2||\Lambda||^4 N < M$ but all points visited by $\pi$ prior to the shortening lie in $\mathbb{N} \times \mathbb{N}_{{\geq} M}$. Altogether we conclude that the shortening of $\pi$ is not only admissible from $\vec{s}$, but even admissible from  $\vec{s}+\vec{v}_1\cdot n\cdot\vec{w}$.

Overall, we have a shortened version of $\rho$ going from $\vec{r}$ to $\vec{s}+\vec{v}_1\cdot n\cdot\vec{w}$ that is admissible. This is followed by a shortened version of $\pi$ going from $\vec{s}+\vec{v}_1\cdot n\cdot\vec{w}$ to $\vec{t}+\vec{v}_1\cdot n\cdot\vec{w}-\vec{w}_1\cdot n\cdot \vec{v}=\vec{t}-n\cdot\tuple{0,\vec{w}_1 \vec{v}_2-\vec{w}_2 \vec{v}_1}$ and which is admissible as well.

Since we successfully shortened $\rho\pi$ by $n\cdot\tuple{0,\vec{w}_1 \vec{v}_2-\vec{w}_2 \vec{v}_1}$ it only remains to show that $\vec{w}_1 \vec{v}_2-\vec{w}_2 \vec{v}_1 \in \{0,\ldots,2 ||\Lambda||^3\}$. Clearly, $\vec{w}_1 \vec{v}_2-\vec{w}_2 \vec{v}_1<\vec{w}_1 \vec{v}_2\le 2||\Lambda||^3$.  On the other hand, it cannot be that
$\vec{v}_1 \vec{w}_2 > \vec{v}_2 \vec{w}_1$,
because it implies
\begin{align*}
\vec{t}_2
& \leq  \tuple{\vec{v}_2 \vec{w}_1, \vec{v}_1 \vec{w}_2} \cdot
        \tuple{-\vec{r}_2,          \vec{t}_2} \\
& =  -\vec{r}_2\vec{v}_2 \vec{w}_1 + \vec{t}_2 \vec{v}_1 \vec{w}_2\\
& =       \vec{v}_2 \cdot \tuple{-\vec{w}_1, \vec{w}_2} \cdot \tuple{\vec{r}_2, -\vec{s}_1} - \vec{w}_2 \cdot  \tuple{-\vec{v}_2, \vec{v}_1} \cdot \tuple{\vec{s}_1, -\vec{t}_2} \\
& <       (\vec{v}_2-\vec{w}_2)\cdot 7(K+2) (M+1) ||\Lambda||^5 \\
& \leq 14(K+2)(M+1) ||\Lambda||^6.\qedhere
\end{align*}

\end{proof}

Our penultimate theorem states that it is not possible
for a shortest reachability witness to visit a point $\vec{f}$ 
whose norm is much larger than the norms of 
the last point close to the axes before visiting $\vec{f}$ and 
the first point close to the axes after visiting $\vec{f}$.

\begin{theorem}
\label{th:far}
Suppose a $2$-\SLPS\ $\Lambda$ with $K$ cycles and with $||\Lambda|| > 0$ has 
a path $\pi$ from point $\vec{s}$ to point $\vec{t}$ such that
\begin{itemize}
\item
all points visited by $\pi$ from $\vec{s}$ are in 
$(\mathbb{N}_{{\geq}\, 6 ||\Lambda||^3})^2$ and
\item
some point $\vec{f}$ visited by $\pi$ from $\vec{s}$ satisfies
\[
        \norm{\vec{f}} \,>\, 3   \norm{\Lambda}^2 \cdot \norm{\{\vec{s}, \vec{t}\}} +
        7.5 \norm{\Lambda}^5 K.
\]
\end{itemize}
There is a shortening of $\pi$ by $\vec{0}$ that is admissible from $\vec{s}$.
\end{theorem}
\begin{proof}
We have that
\begin{align*}
    |\pi| & \geq 2(\norm{\vec{f}} - \norm{\{\vec{s},\vec{t}\}}) / \norm{\Lambda} \\
          & >    4 \norm{\Lambda} \cdot \norm{\{\vec{s},\vec{t}\}} + 15 \norm{\Lambda}^4 K \\
      & \geq 4 \norm{\Lambda} \, \norm{\vec{t} - \vec{s}} + (K 2 \norm{\Lambda}^2 + 1) (4 \norm{\Lambda}^2 + 1).
\end{align*}
In particular, $\vec{C} = \Cycles{2 \norm{\Lambda}^2}{\Lambda}{\pi}$ cannot be empty.
Suppose the cone of $\vec{C}$ does not contain $\vec{0}$.
Then \aref{Lemma}{l:halfplane} provides a vector $\vec{p}$ with
$\norm{\vec{p}}\le 2\norm{\Lambda}$ and
$\vec{p}\cdot \vec{c} > 0$ for all $\vec{c}\in\vec{C}$.
By \aref{Lemma}{l:drift}~(i) we then get
\begin{align*}
4\norm{\Lambda}\norm{\vec{t}-\vec{s}} 
& \ge \vec{p}\cdot(\vec{t}-\vec{s}) \\
& \ge |\pi| - (K 2 \norm{\Lambda}^2+ 1) (4 \norm{\Lambda}^2 + 1),
\end{align*}
which contradicts the inequation above.
So the cone of $\vec{C}$ contains $\vec{0}$
and we finish by \aref{Lemma}{l:cut} with $N = 1$ and $\vec{c} = \vec{0}$.
\end{proof}

We are now equipped to establish that $2$-dimensional \textSLPS s
have pseudo-polynomially bounded reachability witnesses:

\begin{theorem}
\label{thm:2SLPS-short}
Suppose $\Lambda$ is a $2$-\SLPS\ with $K$ cycles.
For any shortest admissible path from $\vec{0}$ to $\vec{0}$,
the norms of all points visited are at most
$2914.5 K \norm{\Lambda}^{15}$.
\end{theorem}
 \begin{proof}
 We can assume $K, \norm{\Lambda} > 0$.
 Consider any shortest admissible $\pi\in\Lambda$ from $\vec{0}$ to~$\vec{0}$,
 and let $M \eqdef 16 \norm{\Lambda}^7$.

 First, we show that at all points visited by $\pi$ where one coordinate
 is less than $M$,
 the other coordinate must be less than
 \begin{multline*}
 M' \eqdef 969 K \norm{\Lambda}^{13} 
         = 19 (3 K)(17\norm{\Lambda}^7)\norm{\Lambda}^{6} \\
       \ge 19 (K+2)(M+1)\norm{\Lambda}^{6}.
 \end{multline*}
 To see this, assume the contrary and let $\vec{t}\in\N^2$ be a point
 visited by $\pi$ from
 $\vec{0}$ such that, w.l.o.g., $\vec{t}_1< M$ and $\vec{t}_2\ge M'$.
 Further assume that $\vec{t}$
 is a point with maximum $\vec{t}_2$ among all points with this property.
\begin{figure}[t]
  \includegraphics[width=0.5\linewidth]{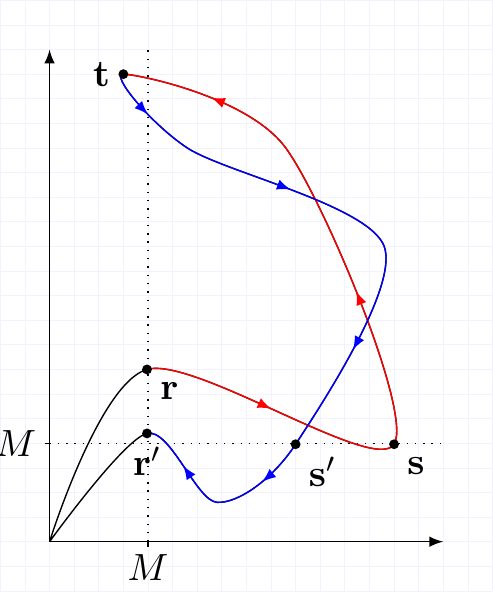}
  \includegraphics[width=0.5\linewidth]{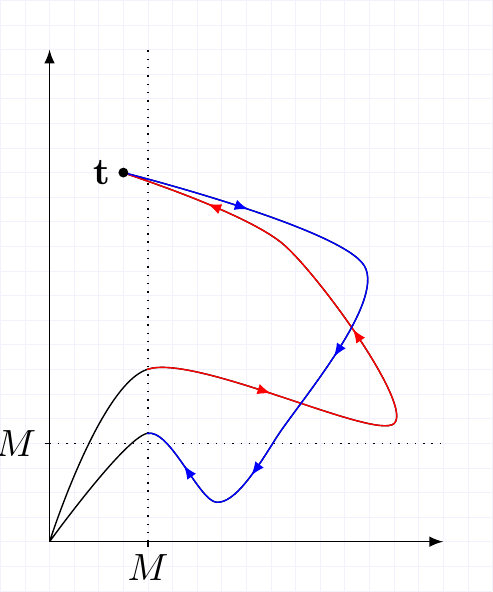}
  \caption{
  In \aref{Theorem}{thm:2SLPS-short}, we identify two path segments, the red one from $\vec{r}$ to $\vec{t}$ via $\vec{s}$
  and the blue one from $\vec{t}$ to $\vec{r'}$ via $\vec{s'}$ pictured on the left. Both can be shortened via \aref{Theorem}{th:one.visit} (where the theorem is applied to the revers of the blue path).
  Both shortenings combined result in a new, shorter path from $\vec{0}$ to $\vec{0}$. In the new path, the point corresponding to $\vec{t}$ has moved down. This is pictured on the right.
  }
  \label{fig:2SLPS-short}
\end{figure}
 Then we can extract a subpath $\rho$
 by following $\pi$ backwards, starting in $\vec{t}$ until for the first time
 a point $\vec{s}$ is visited that satisfies $\vec{s}_2< M$
 and then further, until for the first time a point $\vec{r}$
 is visited with $\vec{r}_1< M$.
 (Here it may be the case that
  $\vec{s}$ and $\vec{r}$ are the same point,
  i.e.\ the latter path segment is empty.)
 On this path \aref{Theorem}{th:one.visit} is applicable
 with $N=2 \norm{\Lambda}^3$. So there exist
 $\gamma\in\{0,\dots 2\norm{\Lambda}^3\}$ and
 shortenings by
 $\tuple{0,n\gamma}$ for all $n\in\{1,\dots N\}$, admissible from the
 point $\vec{r}$.
 If $\gamma = 0$ then this directly contradicts the minimality of $\pi$.
 Otherwise we can, analogously,
 extract a subpath $\rho'$ by following $\pi$ forwards from $\vec{t}$ to some
 $\vec{r'}$ and then reversing, so that \aref{Theorem}{th:one.visit} provides
 $\gamma'\in\{0,\dots 2\norm{\Lambda}^3\}$ and shortenings by
 $\tuple{0,n\gamma'}$ for all $n\in\{1,\dots N\}$, admissible backwards
 from $\vec{r'}$. See \aref{Figure}{fig:2SLPS-short} for an illustration.
 Together, this means there is a shortening of $\pi$ by $\vec{0}$;
 a contradiction with the minimality assumption.

 To show the claim of the theorem, assume that $\pi$ visits some
 point $\vec{f}$ whose norm exceeds
 $2914.5 K \norm{\Lambda}^{15}
  \,\geq\, 3\norm{\Lambda}^2 M' + 7.5\norm{\Lambda}^5 K$.
 Then we can partition $\vec{0}\step{\pi}\vec{0}$ as
  $\vec{0}\step{\rho}
   \vec{s}\step{\sigma}
   \vec{f}\step{\sigma'}
   \vec{t}\step{\tau}
   \vec{0}$
 where $\norm{\vec{s}}, \norm{\vec{t}} < M'$
 and all other points visited by $\sigma \sigma'$ from $\vec{s}$
 are in $(\N_{{\geq} M})^2$.
 But then \aref{Theorem}{th:far}
 provides a shortening of $\sigma \sigma'$ 
 that is admissible from $\vec{s}$,
 and thus a shortening of $\pi$ admissible from $\vec{0}$,
 again contradicting the minimality assumption.
 \end{proof}

\section{Finish: $2$-VASS}
\label{sec:application_to_2_vass}

\citet[Thm.~1]{BFGHM2015} showed that $2$-\VASS\ can be flattened,
i.e., their reachability relation can be expressed by a finite set
of polynomially bounded \textLPS s:
\begin{theorem}
    \label{thm:blondin}
    For every $2$-\VASS\ $V$ with $n$ states
    over an alphabet $\vec{A} \subseteq \Z^2$,
    there exist
    finitely many
    \LPS s 
    $
    \Lambda_1,
    \Lambda_2,
    \dots,
    \Lambda_k
    \subseteq V
    $
    such that
${\step{V}} = \bigcup_{i=1}^k {\step{\Lambda_i}}$
    and
        $|\Lambda_i| \le (\norm{\vec{A}}+n)^{O(1)}$ for all $1\le i\le k$.
\end{theorem}
Small witness theorems for $2$-dimensional \LPS s therefore carry over to $2$-\VASS.
To apply our small witness theorem for \emph{simple} \LPS s
a further reduction (\aref{Theorem}{thm:simple-LPS} below) is necessary.
We will use the following fact.

\begin{lemma}
    \label{lem:vass-loops}
    Suppose $\sys{A}\subseteqfin\Z^d$, $\pi\in\sys{A}^*$,
    $m\in\N$ and $\vec{s}\in\N^d$.
    Then
    $\pi^{m+2}$ is admissible from $\vec{s}$ if and only if
    $\pi(\effect{\pi})^m\pi$ is
    admissible from $\vec{s}$.
\end{lemma}
\begin{proof}
    The `only if' direction is immediate; for the other direction
    observe that if 
    $\pi(\effect{\pi})^m\pi$ is admissible from $\vec{s}$,
    then there is $\vec{t}\in\N^2$ such that
    $\vec{s}\step{\pi(\effect{\pi})^m}\vec{t}$
    and $\pi$ is admissible both from $\vec{s}$ and $\vec{t}$.
    Let $\vec{g}\in\N^d$ be minimal such that $\pi$ is admissible from it.
    We show that $\pi$ is admissible from all points
    $\vec{s}+\effect{\pi}\cdot i$
    for $0\le i \le m$.
    Suppose this fails for some $i$ and $\vec{v} \eqdef
    (\vec{s}+\effect{\pi}\cdot i) \not\ge \vec{g}$.
    Then $\vec{v}_j<\vec{g}_j$ for
    some dimension $1\le j\le d$.
    Since $\vec{s}\ge\vec{g}$,
    it must hold that $(\effect{\pi})_j<0$
    and because $\vec{t} = \vec{s}+\effect{\pi}\cdot(m+1)$,
    also $\vec{t}_j<\vec{g}_j$
    and consequently, $\vec{t}\not\ge\vec{g}$. Contradiction.
\end{proof}

\begin{theorem}
    \label{thm:simple-LPS}
    For every \LPS\ $L$
    there are
    finitely many \SLPS s 
    $
    \Lambda_1,
    \Lambda_2,
    \dots,
    \Lambda_k
    $
    such that
        ${\step{L}} \,=\, \bigcup_{i=1}^k {\step{\Lambda_i}}$
    and
    \begin{enumerate}
        \item For all $i\le k$,
        $|\Lambda_i| \le 4|L|$ and
        $\norm{\Lambda_i} \le 2\norm{L}\cdot|L|$.
        \item For every path $\pi \in \Lambda_i$ there exists
            $\pi'\in L$ with $|\pi'|\le|\pi|\cdot |L|$
                and ${\step{\pi}} \subseteq {\step{\pi'}}$.
        \end{enumerate}
\end{theorem}
\begin{proof}
    The idea is first to split $L$ into a finite set $S$ of \LPS\
    such that each of them predetermines, for each 
    cycle, if it can be used zero, one or more than one times.
    Clearly, $\bigcup S = L$
    and the maximum length of any resulting \LPS\ is $3\len{L}$.
    In each such \LPS\ $\Lambda$ we then replace occurrences
    of subexpressions $\beta_i\beta_i^*\beta_i$
    by subexpressions $\beta_i(\effect\beta_i)^*\beta_i$,
    which does not increase the length and
    can only increase the norm to $\norm{\Lambda}\le \norm{L}\cdot \len{L}$.
    By \aref{Lemma}{lem:vass-loops} this moreover does not
    change the relation $\step{\Lambda}$ and guarantees the second claimed
    property.
    It remains to introduce a total of at most $\len{L}$ many cycles $\vec{0}^*$
    into the unstarred segments to make the \LPS\ simple.
\end{proof}

\begin{theorem}
$2$-VASS have pseudo-polynomially long reachability witnesses.
\end{theorem}

\begin{proof}
Suppose $V$ is a $2$-\VASS\ with $n$ states 
over an alphabet $\vec{A} \subseteq \Z^2$ and
$\vec{s}, \vec{t} \in \N^2$ are such that
$\vec{s} \step{\pi} \vec{t}$ for some path $\pi \in V$.

First note that a $2$-\VASS\ $V' \eqdef (\vec{s})\,V\,(-\vec{t})$,
obtained from $V$ by adding two states,
has an admissible path $\pi'=(\vec{s})\pi(-\vec{t})$ from $\vec{0}$ to~$\vec{0}$.
By Theorems \ref{thm:blondin} and \ref{thm:simple-LPS}
there is an $2$-\SLPS\ $\Lambda$ such that:
\begin{itemize}
\item
$|\Lambda|$ and $\norm{\Lambda}$ are polynomial in 
$n$ and $\norm{\vec{A} \cup \{\vec{s}, \vec{t}\}}$;
\item
$\Lambda$ has an admissible path from $\vec{0}$ to~$\vec{0}$;
\item
for every path $\rho \in \Lambda$, there exists $\rho' \in V'$
with ${\step{\rho}} \subseteq {\step{\rho'}}$ and 
with $|\rho'|$ polynomial in 
$|\rho|$, $n$ and $\norm{\vec{A} \cup \{\vec{s}, \vec{t}\}}$.
\end{itemize}

Now, by \aref{Theorem}{thm:2SLPS-short}, we have 
$\vec{0} \step{\rho} \vec{0}$ for some path $\rho \in \Lambda$
with $|\rho|$ polynomial in $|\Lambda|$ and $\norm{\Lambda}$,
and thus polynomial in $n$ and $\norm{\vec{A} \cup \{\vec{s}, \vec{t}\}}$.
Hence there exists $\rho' \in V'$ 
such that $\vec{0} \step{\rho'} \vec{0}$
and $|\rho'|$ is polynomial in 
$n$ and $\norm{\vec{A} \cup \{\vec{s}, \vec{t}\}}$,
as required.
\end{proof}

A direct consequence is that a nondeterministic algorithm that guesses a bounded witness
on the fly requires space logarithmic in the number of states and the infinity norms of action, source and target vectors.
Recall also that already $0$-VASS are essentially directed graphs.

\begin{corollary}
The reachability problem for $2$-\VASS\ with integers given in unary is \NL-complete.
\end{corollary}

\section{Conclusion} 
\label{sec:conclusion}
That the covering and boundedness problems for \VASS\ given in unary
are \NL-complete for any \emph{fixed} dimension has been known for
thirty years \citep{RY86}.
This contribution suggests that, possibly, 
the same is true for the reachability problem.

If that is too challenging, 
how about restricting to flat \VASS, i.e.\ \textLPS s, 
and attempting to extend the machinery developed here to dimension $3$
in order to close the gap between \NL\ hardness and \NP\ membership 
\citep{BFGHM2015} in that case?

\acks

We are grateful to 
Stefan G\"oller,
Christoph Haase and 
J\'er\^ome Leroux
for helpful conversations.

\bibliographystyle{abbrvnat}
\renewcommand{\doi}[1]{doi: \href{http://dx.doi.org/#1}{#1}}

\bibliography{conferences,journalsabbr,references}

\end{document}